\documentclass[11pt]{article}

\usepackage{amsmath, amsthm, amssymb}
\usepackage{graphicx} 
\usepackage{times}
\usepackage{comment}
\newtheorem{theorem}{Theorem}
\newtheorem{lemma}{Lemma}

\newtheorem{cor}{Corollary}

\author{Vadim V.~Lozin\thanks{Mathematics Institute, University of Warwick, Coventry CV4 7AL, UK. E-mail: V.Lozin@warwick.ac.uk.} 
\and Igor Razgon\thanks{Department of Computer Science and Information Systems,
Birkbeck, University of London, E-mail: igor@dc.bbk.ac.uk}
\and
Viktor Zamaraev\thanks{Mathematics Institute, University of Warwick, Coventry CV4 7AL, UK. E-mail: V.Zamaraev@warwick.ac.uk.}}

\date{}
\title{Well-quasi-ordering does not imply bounded clique-width} 

\begin{document}
\maketitle
\begin{abstract}
We present a hereditary class of graphs of unbounded clique-width which is well-quasi-ordered by the induced subgraph 
relation. This result provides a negative answer to the question asked by Daligault, Rao and Thomass\'e in \cite{DRT10}.
\end{abstract}

\section{Introduction}
Well-quasi-ordering ({\sc wqo}) is a highly desirable property and frequently discovered
concept in mathematics and theoretical computer science \cite{Finkel,Kruskal}. 
One of the most remarkable recent results in this area is the proof of Wagner's conjecture
stating that the set of all finite graphs is well-quasi-ordered by the minor relation
\cite{minor-wqo}. This is, however, not the case for the induced subgraph relation,
since the set of cycles $\{C_n| n \geq 3\}$ forms an infinite antichain with respect to this relation.
On the other hand, the induced subgraph relation may become a well-quasi-order when restricted to  
graphs in particular classes, such as cographs \cite{Damaschke} or $k$-letter graphs \cite{Pet02}.
It is interesting to observe that in both examples  we deal with graphs of bounded clique-width,
which is another property of great importance in mathematics and computer science.
Moreover, the same is true for all available examples of graph classes which are 
well-quasi-ordered by the induced subgraph relation (see e.g. \cite{KL11}).
This raises an interesting question whether the clique-width is always bounded for graphs 
in well-quasi-ordered classes. This question was formally stated as an open problem by Daligault, Rao and Thomass\'e in \cite{DRT10}.
In the present paper, we answer this question negatively by exhibiting a hereditary class of graphs of unbounded clique-width 
which is well-quasi-ordered by the induced subgraph relation.

Our result shows that it is generally non-trivial to determine whether a given problem
definable in Monadic Second Order ({\sc mso}) logic is polynomially solvable on a {\sc wqo} class,  
since unboundedness of clique-width does not allow a straightforward application of Courcelle's theorem \cite{CorMakRotics}. 
This makes the {\sc wqo} classes an interesting object to study from the algorithmic perspective. 
By the way, we are not aware of any result in flavour of \cite{KreutzerT10} applied to clique-width and $MSO_1$, e.g. stating 
that if clique-width of a class is sufficiently large, then the $MSO_1$-model checking is intractable
subject to a widely believed complexity theoretical assumption. Therefore, in light of our result,
the tractability of {\sc mso} model checking on {\sc wqo} classes looks an interesting open question.

Graphs in the class introduced in this paper are \emph{dense} (in particular, they are $P_7$-free). 
The density is a necessary condition, because an earlier result \cite{WQOArxiv} shows that for sparse graph
classes (those where a large biclique is forbidden as a subgraph) well quasi-orderability by induced subgraphs 
imply bounded treewidth (and hence bounded clique-width). We believe that the result of
\cite{WQOArxiv} can be strengthened by showing that well quasi-orderability by induced subgraphs
in sparse classes implies bounded pathwidth (and hence \emph{linear} clique-width \cite{GurskiW05}). 
Our result proved in the present paper shows a stark contrast between dense and sparse graphs in this context.

The rest of the paper is structured as follows. 
In Section~2 we define the class of graphs studied in this paper and state the main result. 
The unboundedness of clique-width and well-quasi-orderability by induced subgraphs is proved in Sections 3 and 4, respectively.
We use standard graph-theoretic notation as e.g. in \cite{Diestel}. The notions
of clique-width and well-quasi-ordering are introduced in respective sections where they are actually used.


\section{The main result} 
\label{mainres}


In this section, we define the class ${\cal D}$, which is the main object of the paper, 
and state the main result.

Let $P$ be a path with vertex set $\{ 1, \ldots, n \}$ with two vertices $i$ and $j$ being adjacent
if and only if $|i-j| = 1$. For vertex $i$, let {\it power} $q(i)$ of $i$ be the largest 
$2^k$ that divides $i$. For example, $q(5)=1,q(6)=2,q(8)=8,q(12)=4$.
Add edges to $P$ that connect $i$ and $j$ whenever $q(i)=q(j)$.
We denote the graph obtained in this way by $D_n$. Figure \ref{DNpic} illustrates graph $D_{16}$. 
\begin{figure}[h]
\centering 
\includegraphics[height=5cm]{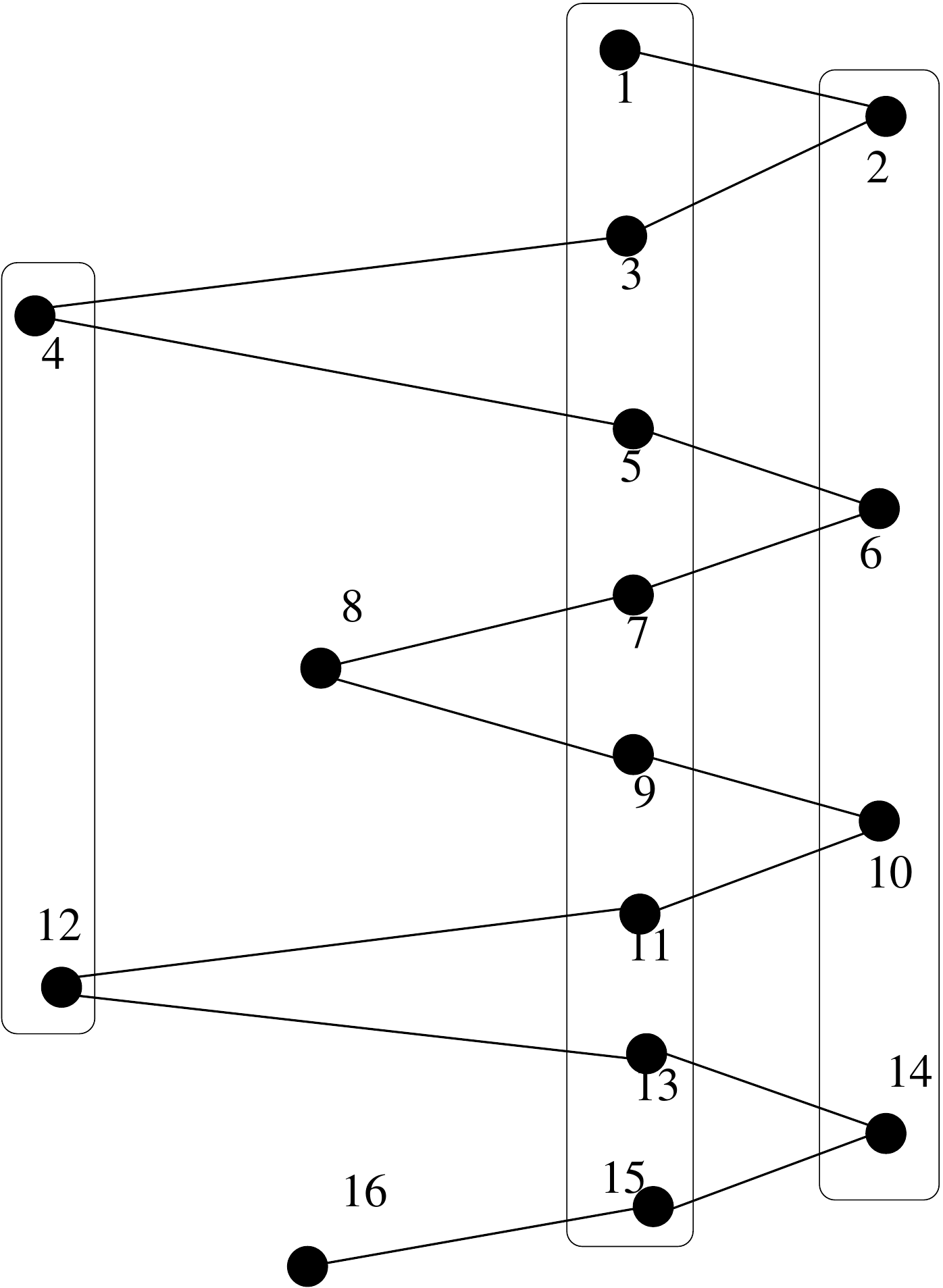}
\caption{Graph $D_{16}$. To avoid shading the picture with many edges, cliques are represented
as rectangular boxes.}
\label{DNpic}
\end{figure}

Clearly, the edges $E(D_n) \setminus E(P)$ form
a set of disjoint cliques and we call them {\it power cliques}. If a power
clique $Q$ contains a vertex $i$ with $q(i)=2^k$ we say that $Q$ \emph{corresponds to} $2^k$.
We call $P$ the \emph{body} of $D_n$, the edges of $E(P)$
the \emph{path edges}, and the edges of $E(D_n) \setminus E(P)$ the \emph{clique edges}.
The class ${\cal D}$ is the set of all graphs $D_n$ and all their induced subgraphs. 
In what follows we prove that 
\begin{itemize}
\item clique-width of graphs in ${\cal D}$ is unbounded  (Section~\ref{nbunbound}),
\item graphs in ${\cal D}$ are well-quasi-ordered by the induced subgraph relation (Section~\ref{secwqo}).
\end{itemize}

These two facts imply the following conclusion, which is the main result of the paper.   

\begin{theorem}
Within the family of hereditary graph classes, there exist classes of 
unbounded clique-width which are well-quasi-ordered by the induced subgraph 
relation.
\end{theorem}


\section{Clique-width is unbounded in ${\cal D}$} 
\label{nbunbound}


The clique-width of a graph $G$, denoted $cwd(G)$, is the minimum number of labels needed to construct the graph 
by means of the four graph operations: creation of a new vertex, disjoint union of two labeled graphs, 
connecting vertices with specified labels $i$ and $j$, and renaming label $i$ to label $j$. 
Every graph $G$ can be constructed by means of these four operations, and the process of the construction 
can be described either by an algebraic expression or by a rooted binary tree, 
whose leaves correspond to the vertices of $G$, the root corresponds to $G$ and
the internal nodes correspond to the union operations.

Given a graph $G$ and a subset $U\subset V(G)$, we denote by $\overline{U}$ the set $V(G)-U$.
We say that two vertices $x,y\in U$ are $U$-{\it similar} if $N(x)\cap \overline{U}=N(y)\cap \overline{U}$,
i.e. if $x$ and $y$ have the same neighbourhood outside of $U$. Clearly, the $U$-similarity is an equivalence 
relation and we denote the number of similarity classes of $U$ by $\mu_G(U)$. Also, we denote 
$$
\mu(G)=\min\limits_{\frac{1}{3}n \leq |U| \leq \frac{2}{3}n} \mu_G(U),
$$
where $n=|V(G)|$. 
Our proof of the main result of this section is based on the following lemma.

\begin{lemma} \label{cwdnb}
For any graph $G$, $\mu(G) \leq cwd(G)$.
\end{lemma}

\begin{proof}
Let $T$ be an {\it optimal} decomposition tree, $t$ a node of $T$ and $U_t$ the set of vertices of $G$ that are leaves 
of the subtree of $T$ rooted at $t$. It is known (see e.g. \cite{rcw}) that $cwd(G)\ge \mu_G(U_t)$ for any node $t$ of $T$.
According to a well known folklore result, the binary tree $T$ has a node $t$ such that 
$\frac{1}{3}|V(G)| \leq |U_t| \leq \frac{2}{3}|V(G)|$, in which case $\mu_G(U_t)\ge \mu(G)$. Hence the lemma.
\qed \end{proof}

\medskip

Let $U \subseteq V(D_n)$, and let $P$ be the body of $D_n$. We denote by $P^U$ the 
subgraph of $P$ induced by $U$. In other words, 
$P^U$ is obtained from $D_n[U]$ by removing the clique edges. Since $P$ is a path,
$P^U$ is a graph every connected component of which is a path. 

\begin{lemma} 
\label{manysplits3}
If $P^U$ has $c+1$ connected components, then $\mu_{D_n}(U) \geq c/2$. 
\end{lemma}

\begin{proof}
In the $i$-th connected component of $P^U$, $i \leq c$, we choose the last vertex (listed along the path $P$)
and denote it by $u_i$. The next vertex of $P$, denoted $\overline{u}_i$, belongs to $\overline{U}$. 
This creates a matching of size $c$ with edges $(u_i,\overline{u}_i)$. 
Note that none of $(u_i,\overline{u}_j)$ is a path edge for $i < j$.
Among the chosen vertices of $U$ at least half have the same parity.
Their respective matched vertices of $\overline{U}$ have the opposite parity. Since the clique edges connect 
only the vertices of the {\it same} parity, we conclude that at least $c/2$ vertices of $U$ have pairwise 
different neighbourhoods in $\overline{U}$, i.e.  $\mu_{D_n}(U) \geq c/2$.
\qed \end{proof}

Note that if  $P^{U}$ has $c$ connected components, then  $P^{\overline{U}}$  has at
least $c-1$ connected components. Therefore, in light of Lemma \ref{manysplits3}, it remains to consider
the case where both $P^{U}$ and $P^{\overline{U}}$ have a limited number of connected components. By Lemma \ref{cwdnb}
we can assume that both $U$ and $\overline{U}$ are `large', and hence each of $P^{U}$ and $P^{\overline{U}}$
has a `large' connected component.
In order to address this case we use the following lemma which states that a large number of power cliques
intersecting both $U$ and $\overline{U}$ implies a large value of $\mu_{D_n}(U)$.

\begin{lemma} \label{manysplits}
If there exist $c$ different power cliques $Q_1, \ldots, Q_c$ each of which
\begin{enumerate}
	\item[(1)] corresponds to a power of 2 greater than 1 and
	\item[(2)] intersects both $U$ and $\overline{U}$
\end{enumerate}
then $\mu_{D_n}(U) \geq c$.
\end{lemma}

\begin{proof}
	Let $u_i$ and $\overline{u}_i$ be some vertices in $Q_i$, which belong to
	$U$ and $\overline{U}$, respectively. Since all the vertices in 
	$M = \{ u_1, \overline{u}_1, \ldots, u_c, \overline{u}_c \}$ are even and two
	even vertices are adjacent in $D_n$ if and only if they belong to the same
	power clique, $M$ induces a matching in $D_n$ with edges 
	$(u_i, \overline{u}_i)$, $i= 1 \ldots, c$. This implies that $u_1, \ldots, u_c$
	have pairwise different neighbourhoods in $\overline{U}$, that is 
	$\mu_{D_n}(U) \geq c$.
\qed \end{proof}

The only remaining ingredient to prove the main result of this section is the following lemma.

\begin{lemma}\label{manybig}
	Let $c$ be a constant and $P'$ a subpath of $P$ of length at least $2^{c+1}$.
	Then $P'$ intersects each of the power cliques corresponding to $2^1, \dots, 2^c$.
\end{lemma}

\begin{proof}
	The statement easily follows from the fact that for a fixed $k$ vertices $v$ with 
	$q(v) = 2^k$ are of the form $v = 2^k(2p+1)$. That is, they
	occur in $P$ with period $2^{k+1}$.
\qed \end{proof}

Now we are ready to prove the main result of this section. 

\begin{theorem} \label{nbub}
	Let $n$ and $c$ be natural numbers such that $n \geq 3((2c+1)(2^{c+1}-1)+1)$. Then
	$cwd(D_n) \geq c$ and hence the clique-width of graphs in $\cal D$ is unbounded.
\end{theorem}

\begin{proof}
	Let $U$ be an arbitrary subset of vertices of $D_n$, such that $\frac{n}{3} \leq |U| \leq \frac{2n}{3}$.
	Note that the choice of $U$ implies that the cardinalities of both $U$ and $\overline{U}$ are at least
	$\frac{n}{3} \geq (2c+1)(2^{c+1}-1)+1$.
	
	If $P^U$ has at least $2c+1$ connected components, then by Lemma \ref{manysplits3}
	$\mu_{D_n}(U) \geq c$. Otherwise $P^U$ has less than $2c+1$ connected components 
	and $P^{\overline{U}}$ has less than $2c+2$
	connected components. By the pigeonhole principle, both graphs have connected components
	of size at least $2^{c+1}$. Clearly, these connected components are disjoint subpaths of $P$.
	By Lemma~\ref{manybig}, the power cliques corresponding to $2^1, \ldots, 2^c$ intersect
	both $U$ and $\overline{U}$, and hence, by Lemma~\ref{manysplits}, $\mu_{D_n}(U) \geq c$.
	
	Since $U$ has been chosen arbitrarily, we conclude that $\mu(D_n) \geq c$, and therefore, by 
	Lemma~\ref{cwdnb}, $cwd(D_n) \geq c$, as required.
\qed \end{proof}


\section{${\cal D}$ is WQO by induced subgraphs} 
\label{secwqo}

A binary relation $\le$ on a set $W$ is a {\it quasi-order} (also known as {\it preorder}) 
if it is reflexive and transitive. Two elements $x,y \in W$ are said to be comparable 
with respect to $\le$ if either $x\le y$ or $y\le x$. Otherwise, $x$ and $y$ are incomparable.
A set of pairwise comparable elements is called a {\it chain} and a set of pairwise incomparable elements
an {\it antichain}. A quasi-order $(W,\le)$ is a {\it well-quasi-order} ({\sc wqo}) if it contains neither 
infinite strictly decreasing chains nor infinite antichains.

In this section, we show that graphs in ${\cal D}$ are well-quasi-ordered by the induced subgraph relation.
In the proof we apply the celebrated Higman's lemma \cite{Hig52} which can be stated as follows.

For an arbitrary set $M$, let $M^*$ be the set of all finite sequences of elements of $M$.
Any quasi-order $\le$ on $M$ defines a quasi-order $\preceq$ on $M^*$ as follows:
$(a_1,\ldots,a_m)\preceq (b_1,\ldots,b_n)$ if and only if there is an order-preserving injection 
$f:\ \{a_1,\ldots,a_m\}\to \{b_1,\ldots,b_n\}$ with $a_i\le f(a_i)$ for each $i=1,\ldots,m$. 

\begin{lemma}\label{HL} \cite{Hig52}
If $(M,\le)$ is a {\sc wqo}, then $(M^*,\preceq)$ is a {\sc wqo}.
\end{lemma}  



Obviously, the induced subgraph relation contains no infinite strictly decreasing chains. Therefore, 
to prove that this relation is a {\sc wqo} on ${\cal D}$ we need to 
show that for each  infinite sequence ${\cal G}=G_1,G_2 \dots$ of graphs in ${\cal D}$
there are $i,j$ such that $G_i$ is an induced subgraph of $G_j$. 

We recall that $V(D_n)$ is the set of integers $1,2,\ldots,n$ listed along the body of $D_n$
and any graph in $\cal D$ is an induced subgraph of $D_n$ with some $n$.
Among all possible sets of integers inducing a graph (isomorphic to) $G\in \cal D$ we pick \emph{one}
(arbitrarily) and identify $G$ with this set.  

Any set of consecutive integers will be called an {\it interval} and any graph in $\cal D$ 
induced by an interval will be called a {\it factor}. The number of elements in an interval inducing 
a factor is called the {\it length} of the factor. If a graph $G\in \cal D$ is not 
a factor, its vertex set can be split into maximal intervals and we call the subgraphs of $G$
induced by these intervals {\it factor-components}  of $G$. 
The set of all factor-components of $G$ will be denoted ${\cal F}(G)$.

\begin{lemma} \label{longpaths}
If ${\cal G}$ contains graphs with arbitrarily long factor-components, then ${\cal G}$ is not an antichain.
\end{lemma}

\begin{proof}
Pick an arbitrary $G_i$ and let $n$ be the smallest number such
that $G_i$ is an induced subgraph of $D_n$.
By our assumption, there is $G_j$ with factor-component $F$ of length at least $5n$.
Let us show that $D_n$ is an induced subgraph of $G_j$.
By the transitivity of the induced subgraph relation, this will imply
that $G_i$ is an induced subgraph of $G_j$.


Let $2^k$ be the smallest power of $2$ larger than $n$. Clearly, $2^{k+1} \leq 4n$.
Hence, by Lemma \ref{manybig}, there is a vertex $y$ among the first $4n$ vertices  
of $F$ with $q(y)=2^k$.
Let $F'$ be the factor induced by the vertices of $F$ starting at $y+1$. 
Since $F$ is of length at least $5n$ and $y$
is among the first $4n$ vertices of $F$, the length of $F'$ is at least $n$. 
Thus we can define an injective function $f:V(D_n) \rightarrow V(F')$ as follows:
$f(z)=y+z$ for $1 \leq z \leq n$. We claim that $f$ is an induced subgraph isomorphism
from $D_n$ to a subgraph of $G_j$. Clearly, $f(z+1)=f(z)+1$ for $1 \leq z<n$, hence it remains to
verify that adjacencies and non-adjacencies are preserved for vertices $z_1,z_2$ of
$D_n$ such that $z_2>z_1+1$. Clearly, in this case $z_1$ and $z_2$ are adjacent if and 
only if $q(z_1)=q(z_2)$. Moreover, since $f(z_2)>f(z_1)+1$, $f(z_2)$ and $f(z_1)$ are
adjacent if and only if $q(f(z_1))=q(f(z_2))$. Below we prove that  $q(f(z))=q(z)$ for $1 \leq z \leq n$
and hence $q(z_1)=q(z_2)$ if and only if $q(f(z_1))=q(f(z_2))$, implying the lemma.

Indeed, $f(z)=y+z=2^k p+2^{k_1} p_1$, where $2^{k_1}=q(z)$ and $p,p_1$ are odd numbers.
Since $2^{k_1} \leq n <2^k$, $k_1<k$ and hence $y+z$ can be written as 
$2^{k_1} (2^{k-k_1} p+p_1)$.
Since $k>k_1$, $2^{k-k_1}$ is even and hence $2^{k-k_1} p+p_1$ is odd. Consequently,
$q(y+z)=2^{k_1}$, as required. 
\qed \end{proof}

From now on, we assume the length of factor-components of graphs in ${\cal G}$ is bounded by some constant $c=c({\cal G})$.
In what follows we prove that in this case ${\cal G}$ is not an antichain as well.

$\medskip$

Let $F$ be a factor. We say that a vertex $u$ of $F$ is {\it maximal} if $q(u) \geq q(v)$ for each vertex $v$ of $F$ different from $u$.

\begin{lemma} \label{maxone}
Every factor $F$ of $D_n$ contains precisely one maximal vertex.
\end{lemma}

\begin{proof}
Suppose that $F$  contains two maximal vertices $2^kp$ and $2^k(p+r)$
for some odd number $p$ and even number $r \geq 2$.
Then $F$ also contains the vertex $2^k(p+1)$. Clearly $p+1$ is an even
number and hence $q(2^k(p+1)) \geq 2^{k+1}$, which contradicts the
maximality of $2^k$. 
\qed \end{proof}

In light of Lemma \ref{maxone}, we denote the unique maximal vertex
of $F$ by $m(F)$. Also, let $s(F)$ be the smallest vertex of $F$.

$\medskip$

Now we define two equivalence relations on the set of factor graphs as follows. 
We say that two factors $F_1$ and $F_2$ are 
\begin{itemize}
	\item $t$-equivalent if they are of the same length and $m(F_1)-s(F_1)=m(F_2)-s(F_2)$,
	\item $\ell$-equivalent if $q(m(F_1))=q(m(F_2))$. 
\end{itemize}
We denote by $L_i$ the $\ell$-equivalence class such that $q(m(F))=2^i$ for every factor $F$ in this class.
We also order the $t$-equivalence classes (arbitrarily) and denote by $T_j$ the $j$-th class in this order.

\begin{lemma} \label{diffq}
	Let $F$ be a factor of length at most $c$. Let $v$ be a vertex of $F$ different from
	its maximal vertex $m = m(F)$. Then $q(v) = q(|m-v|)$ and, in particular, $q(v) < c$.
\end{lemma}
\begin{proof}
	We can assume without loss of generality that $v>m$.
	Let $k_1,p_1,k_2,p_2$ be such that $m=2^{k_1}p_1$
	and $v-m=2^{k_2}p_2$, with $p_1,p_2$ being odd numbers.
	Observe that $k_2<k_1$. Indeed, otherwise
	$v=2^{k_1}p_1+2^{k_2}p_2=2^{k_1}(p_1+2^{k_2-k_1}p_2)$,
	where $p_1+2^{k_2-k_1}p_2$ is a natural number. 
	Therefore, $q(v) \geq 2^{k_1}=q(m)$ in 
	contradiction either to the maximality of $m$ or to
	Lemma \ref{maxone}. 

	Consequently,
	$v=2^{k_1}p_1+2^{k_2}p_2=2^{k_2}(2^{k_1-k_2}p_1+p_2)$,
	where $2^{k_1-k_2}p_1+p_2$ is an odd number because of
	$2^{k_1-k_2}p_1$ being even. Hence,
	$q(v)=2^{k_2}=q(v-m)$.

	Finally, since the length of $F$ is at most $c$, we conclude that $v-m<c$, and 
	therefore $q(v)=q(v-m)<c$.
\qed \end{proof}

\begin{cor} \label{maxpower}
	Let $F$ be a factor of length at most $c$. Let $m$ be a vertex of $F$ with 
	$q(m) \geq c$. Then $m$ is the maximal vertex of $F$.
\end{cor}

\begin{cor} \label{intiso}
	Let $F_1, F_2$ be two $t$-equivalent factors. 
	Then there exists an isomorphism
	$f$ from $F_1$ to $F_2$ such that:
	\begin{enumerate}
		\item[(a)] $f(m(F_1)) = m(F_2)$;
		\item[(b)] $q(f(v)) = q(v)$ for all $v \in V(F_1)$ except possibly for $m(F_1)$.
	\end{enumerate}
\end{cor}
\begin{proof}
	We claim that the function $f$ that maps the $i$-th vertex of factor $F_1$ 
	(starting from the smallest) to
	the $i$-th vertex of factor $F_2$ is the desired isomorphism. Indeed, property (a)
	follows from the condition that the factors are $t$-equivalent. Now property (a) 
	together with Lemma \ref{diffq} implies property (b). Finally, since adjacency between
	vertices in a factor is completely determined by their adjacency in the
	body and by their powers, we conclude that $f$ is, in fact, isomorphism. 
\qed \end{proof}

For a graph $G \in \cal D$, we denote by $G_{i,j}$ the set of factor-components of $G$ in $L_i \cap T_j$,
and define a binary relation $\le$ on graphs of $\cal D$ as follows: $G\le H$ if and only if $|G_{i,j}|\le |H_{i,j}|$
for all $i$ and $j$ (clearly in this definition one can be restricted to non-empty sets $G_{i,j}$).

Finally, for a constant  $c=c({\cal G})$ we slightly modify the definition of $\le$ to $\le _c$ as follows.
We say that a mapping $h: \mathbb{N} \to \mathbb{N}$ is $c$-preserving if it is injective and 
$h(i)=i$ for all $i \leq \lfloor \log c \rfloor$. Then   
 $G\le_c H$ if and only if there is a $c$-preserving mapping $h$ such that  $|G_{i,j}|\le |H_{h(i),j}|$
for all $i$ and $j$.

The importance of the binary relation $\leq_c$ is due to the following lemma.

\begin{lemma} \label{indorder}
Suppose the length of factor-components of  $G$ and $H$ is bounded by $c$ and 
$G \leq_c H$, then $G$ is an induced subgraph of $H$. 
\end{lemma}

\begin{proof}
	We say that a factor $F$ is {\it low-powered} if $F \in L_i$, for some $i \leq \lfloor \log c \rfloor$,
        i.e. $q(m(F))\le c$.

	It can be easily checked that the definition of $\leq_c$ implies the existence of an injective function 
	$\phi : {\cal F}(G) \rightarrow {\cal F}(H)$ that possesses the following properties:
	\begin{enumerate}
		\item[(1)] $\phi$ maps each of the factors in ${\cal F}(G)$ to a $t$-equivalent factor in 
		${\cal F}(H)$;
		
		\item[(2)] $F \in {\cal F}(G)$ is a low-powered factor if and only if $\phi(F)$ is;
		
		\item[(3)]  $\phi$ preserves power of the maximal vertex for each of the 
		low-powered factors, i.e. $q(m(F)) = q(m(\phi(F)))$ for every low-powered factor 
		$F \in {\cal F}(G)$;
		
		\item[(4)] for any two factors $F_1, F_2 \in {\cal F}(G)$, $q(m(F_1)) = q(m(F_2))$
		if and only if $q(m(\phi(F_1))) = q(m(\phi(F_2)))$.
	\end{enumerate}
	
	To show that $G$ is an induced subgraph of $H$ we define a witnessing function
	that maps vertices of a factor $F \in {\cal F}(G)$ to vertices of $\phi(F) \in {\cal F}(H)$ 
	according to an
	isomorphism described in Corollary \ref{intiso}. This mapping guarantees that 
	a factor $F$ of $G$ is isomorphic to the factor $\phi(F)$ of $H$. Therefore it remains to 
	check that adjacency relation between 
	vertices in different factors is preserved under the defined mapping.
	
	Note that adjacency between two vertices in different factors is
	determined entirely by powers of these vertices. Moreover, Corollary \ref{intiso} 
	and property (3) of $\phi$ imply that our mapping preserves powers of all vertices 
	except possibly 
	maximal vertices of power more that $c$. Therefore in order to complete
	the proof we need only to make sure that in graph $G$ a maximal vertex $m$ of a 
	factor $F$ with $q(m)>c$ is adjacent to a vertex $v$ in a factor different from $F$
	if and only if the corresponding images of $m$ and $v$ are adjacent in $H$.
	
	Taking into account Corollary \ref{maxpower} we derive that a maximal vertex 
	with $q(m) > c$ is adjacent to a vertex $v$ in a different factor if and only if 
	$v$ is maximal and $q(m)=q(v)$. Now the desired conclusion follows from
	Corollary \ref{intiso} and properties (2) and (4) of function $\phi$.
\qed \end{proof}

\begin{lemma}\label{short}
The set of graphs in $\cal D$ in which factor-components have size at most $c$ is well-quasi-ordered
by the $\le_c$ relation. 
\end{lemma}

\begin{proof}
We associate with each graph $G\in \cal D$ containing no factor-component of size larger than $c$ 
a matrix $M_G=m(i,j)$ with $m(i,j)=|G_{i,j}|$.

Each row of this matrix corresponds to an $\ell$-equivalence class
and we delete any row corresponding to $L_i$ with $i > \lfloor \log c \rfloor$ which is empty (contains only $0$s).
This leaves a finite amount of rows (since $G$ is finite).

Each column of $M_G$ corresponds to a $t$-equivalence class and we delete all columns corresponding to $t$-equivalence classes 
containing factors of size larger than $c$ (none of these classes has a factor-component of $G$). This leaves precisely $\binom{c+1}{2}$ columns in $M_G$. 
 
We define the relation $\preceq_c$ on the set of matrices constructed in this way as follows. We say that $M_1\preceq_c M_2$ 
if and only if there is a $c$-preserving mapping $\beta$ such that  $m_1(i,j)\le m_2(\beta(i),j)$ for all $i$ and $j$. 

It is not difficult to see that if $M_{G_1}\preceq_c M_{G_2}$, then $G_1\le_c G_2$. Therefore, if $\preceq_c$ is a well-quasi-order,
then $\le_c$ is a well-quasi-ordered too. The well-quasi-orderability of matrices follows by repeated applications of Higman's lemma.
First, we split each matrix $M$ into two sub-matrices $M'$ and $M''$ so that $M'$ contains the first $\lfloor \log c \rfloor$ rows  and $M''$
contains the remaining rows. 

To see that the set of matrices $M'$ is {\sc wqo} we apply Higman's lemma twice. 
First, the set of rows is {\sc wqo} since each of them is a finite word over the alphabet of non-negative integers 
(which is {\sc wqo} by the ordinary arithmetic $\le$ relation). Second, the set of matrices is {\sc wqo} since each of them 
is a finite word over the alphabet of rows. 

Similarly, the set of  matrices $M''$ is {\sc wqo}.

Note that in both applications of Lemma \ref{HL} to $M'$ and in the first application to $M''$,
we considered sets of sequences of the \emph{same length}. Hence, in this, case, Higman's
lemma in fact implies the existence of two sequences one of them is \emph{coordinate-wise}
smaller than the other, exactly what we need in these cases. 


Finally, the set of matrices $M$ is {\sc wqo} since each of them is a word of two letters 
($M'$ and $M''$) over the alphabet which is {\sc wqo}. 
\qed \end{proof}

Combining Lemmas~\ref{longpaths} and \ref{short}, we obtain the main result of this section.

\begin{theorem} \label{mainwqo}
${\cal D}$ is {\sc wqo} by the induced subgraph relation.
\end{theorem}

\end{document}